\documentclass[11pt, letterpaper]{article}
\usepackage[T1]{fontenc}
\usepackage[margin=0.8in]{geometry}
\usepackage{amssymb}
\usepackage{graphicx}
\usepackage{algorithm}
\usepackage[noend]{algorithmic}
\usepackage{microtype}
\usepackage{subfig}
\usepackage{babel}
\usepackage{amsmath}
\usepackage{amsthm}
\usepackage{enumerate}
\usepackage{wrapfig}
\usepackage{multirow}
\usepackage{nicefrac}
\usepackage[colorinlistoftodos,prependcaption,textsize=tiny]{todonotes}
\usepackage[colorlinks,linkcolor=blue,filecolor=blue,citecolor=blue,urlcolor=blue,pdfstartview=FitH,pagebackref]{hyperref}
\usepackage[nameinlink]{cleveref}
\usepackage{authblk}
\usepackage{thm-restate}
\usepackage{multicol}
\usepackage[T1]{fontenc}
\usepackage{lmodern}

\usepackage{url}

\definecolor{forestgreen}{rgb}{0.13, 0.55, 0.13}

\newtheorem{theorem}{Theorem}
\newtheorem{lemma}[theorem]{Lemma}

\newtheorem{observation}[theorem]{Observation}

\newtheorem{definition}[theorem]{Definition}

\graphicspath{{./figures/}}

\def\B{\mathcal{B}}
\def\C{\mathcal{C}}

\def\F{\mathcal{F}}

\def\reals{{\mathbb R}}

\def\R{\mathbb{R}}

\def\eps{{\varepsilon}}

\newcommand{\frechet}{Fr\'echet}

\date{}

\newcommand{\old}[1]{{{}}}
\newcommand{\doubleBlind}[1]{{{}}}
\newcommand{\SDcenter}[1]{{{}}}

\newcommand{\new}[1]{{{}}}

\date{} 

\title{On Flipping the Fr\'{e}chet distance}

\author[1]{Omrit Filtser}
\author[2]{Mayank Goswami\thanks{This work is supported by US National Science Foundation (NSF) awards CRII-1755791 and CCF-1910873.}}
\author[3]{Joseph S.B.Mitchell\thanks{This work is partially supported by the National Science Foundation (CCF-2007275), the US-Israel Binational Science Foundation (BSF project 2016116), Sandia National Labs, and DARPA (Lagrange).}}
\author[4]{Valentin Polishchuk}
	
 \affil[1]{The Open University of Israel\\ \texttt{omrit.filtser@gmail.com}}
 \affil[2]{Queens College CUNY\\ \texttt{mayank.goswami@qc.cuny.edu}}
 \affil[3]{Stony Brook University\\ \texttt{joseph.mitchell@stonybrook.edu}}
 \affil[4]{Link\"{o}ping University\\ \texttt{valentin.polishchuk@liu.se}}

\begin{document}
	
	\maketitle
	
	%\pagenumbering{} 
	
	\begin{abstract}
		
	The classical and extensively-studied \textit{Fr\'echet distance} between two curves is defined as an \emph{inf max}, where the infimum is over all traversals of the curves, and the maximum is over all concurrent positions of the two agents. In this article we investigate a ``flipped'' Fr\'echet measure defined by a \emph{sup min} -- the supremum is over all traversals of the curves, and the minimum is over all concurrent positions of the two agents. This measure produces a notion of ``social distance'' between two curves (or general domains), where agents traverse curves while trying to stay as far apart as possible. 

    We first study the flipped Fr\'echet measure between two polygonal curves in one and two dimensions, providing conditional lower bounds and matching algorithms. We then consider this measure on polygons, where it denotes the minimum distance that two agents can maintain while restricted to travel in or on the boundary of the same polygon. We investigate several variants of the problem in this setting, for some of which we provide linear time algorithms. Finally, we consider this measure on graphs. 
    
    We draw connections between our proposed flipped Fr\'echet measure and existing related work in computational geometry, hoping that our new measure may spawn investigations akin to those performed for the Fr\'echet distance, and into further interesting problems that arise.	
		
	\end{abstract}
	
	\vfill
	
	\setlength{\columnsep}{30pt} 
	\begin{multicols}{2}
	{\small \setcounter{tocdepth}{2} \tableofcontents}
	\end{multicols}
	\vfill 
	
	\newpage

	\pagenumbering{arabic} 
	\section{Introduction}
The classical Fr\'{e}chet distance between two curves $P$ and $Q$ is defined as the minimum length of a leash required for a person to walk their dog, with the person and the dog traversing $P$ and $Q$ from start to finish, respectively. Inspired by the challenge of maintaining social distancing among groups and individuals, we consider the question of developing a notion opposite to the \frechet\ distance, where instead of keeping the agents close (short leash), we keep them as far apart as possible.

In this paper we propose a new measure, called the \emph{Flipped \frechet} measure, to capture the amount of social distancing possible while traversing two curves. While \frechet\ distance is defined as an \emph{inf max}, where the infimum is over all traversals of the curves, and the maximum is over all concurrent positions of the two agents, the flipped \frechet\ measure\footnote{One observes that this measure is not a metric/distance as it does not satisfy the triangle inequality.} is defined as a \emph{sup min} -- the supremum is over all traversals of the curves, and the minimum is over all concurrent positions of the two agents. How efficiently can this measure be computed, for curves in one or two dimensions? What if the two agents are walking on edges of a graph, which may or may not be embedded in the plane? Such questions have been considered for \frechet\ distance, and in this paper we initiate their study for the flipped \frechet\ measure.

We refer to the two agents as ``Red'' and ``Blue'' henceforth. Considering the social distancing problem further, what if Blue is not restricted to move along some given curve; rather, it can choose its own path? We now start arriving at a class of problems that have no analogues in the \frechet\ version. Of course, if Blue had no restrictions at all, it could just go to infinity and thus be far from Red (on any path). It therefore makes sense to restrict the domain for Blue, e.g. to a simple polygon $P$ in which Red is traveling, and measure separation using geodesic distance in $P$. We consider questions regarding the complexity of calculating a strategy for Blue to stay away from Red, when Red is traveling on a given path, which may or not be a geodesic in $P$.

An architect designing spaces within a building is faced with choices about the shapes of these spaces, where one must choose, say, between two polygons $P$ and $Q$ where agents will move, in hopes to maximize the potential for social distancing. In order to do so, it would be useful to have a notion of a ``social distancing width'' of a polygon, that captures the difficulty or ease with which two agents can move around in a polygon while maintaining separation. Consider a simple polygon $P$ where the Red agent is on a mission to follow a path, e.g. to traverse the boundary of $P$, while the Blue agent moves within $P$ (with a starting point of Blue's choice), in order to maximize the minimum Red-Blue distance. We define the {\em social distance width} (SDW for short) of a polygon $P$ to be the minimum Red-Blue distance that can be maintained throughout the movement, maximized over all possible movement strategies, and study algorithms to compute the SDW of a polygon. 

For all of the above problems, in addition to developing algorithms for the general versions, we also consider special scenarios which facilitate faster algorithms; for example, while our algorithm for computing the SDW for general polygons runs in quadratic time, we show that for skinny polygons (or a tree), one can compute the SDW in linear time.

\SDcenter{
Next, consider a polygonal domain, possibly with obstacles. Is there a notion of the ``safest location'' inside this region? Intuitively, this would be a point that maximizes the social distancing to traffic (agents moving around) within this region. Assuming the agents follow geodesic paths, we state this notion precisely, and give a polynomial-time algorithm to compute it. 
}

Although this article mostly considers the above problems in the case of $k=2$ agents, in general one may be given $k$ agents and $k$ associated domains. Each agent is restricted to move only within its respective domain, and at least one of the agents has some {\em mission}, e.g., to move from a given start point to a given end point, or to traverse a given path inside the domain. In addition, the domains may be shared or distinct, and different agents may have different speeds. The goal is to find a movement strategy for all the agents, such that the \emph{minimum pairwise} distance between the agents at any time is maximized. Additionally, one may seek to minimize the time necessary to complete one or more missions. 

This new class of problems is different from the usual motion planning problems between robots, or disjoint disks, in some fundamental aspects. Most, if not all, literature on robot motion planning assumes robots are cooperating on some task. One then considers optimizing objectives like makespan, or total distance travelled, etc. However, the kind of movement we consider is far from cooperative -- in fact, some agents may not care about social distancing, while others do. Some may be ``on a mission'' while others are just trying to maintain a safe distance. In addition, different agents may have different starting times and deadlines. Furthermore, as mentioned above, one also encounters design problems, where one may want to configure a layout of a building, a floor plan, or designate rules for traffic flow, in order to facilitate social distancing.

\vspace{5pt}
\noindent\textbf{Related Work.}
%\paragraph{Related Work.}
The \frechet\ distance is an extensively investigated distance measure for curves, starting with the early work of Alt and Godau~\cite{AG95} in '95. There is a quadratic-time algorithm for computing it~\cite{EM94,AG95}, and it was recently shown~\cite{Bringmann14} that under the Strong Exponential Time Hypothesis (SETH), no subquadratic algorithm exists, not even in one dimension~\cite{BM16}. Moreover, under SETH, no subquadratic algorithm exists for approximating the \frechet\ distance within a factor of 3~\cite{BOS19}. 

The problem of coordinating collision-free motion of two agents traveling on polygonal curves was considered already in '89 by O`Donnell and Lozano-Perez \cite{OL89}, in the context of robot manipulators. Assuming some additional restrictions on the movements of the agents (e.g. robots are not allowed to simultaneously traverse segments that are too close), they give an $O(n^2\log n)$ algorithm for minimizing the completion time.

There is an extensive literature on related problems of motion planning in robotics. Perhaps most closely related to our work is that of coordinated motion planning of 2 or more disks; see \cite{DFKMS19}, on nearly optimal (in terms of lengths of motions) rearrangements of multiple unit disks and the related work of \cite{HH02}.
(In our problems, instead of minimizing length of motion for given radius disks, we seek to maximize the radii.)

The problem of computing safe paths for multiple speed-bounded mobile agents that must maintain separation standards arises in air traffic management (ATM) and Aircraft/Train Scheduling applications. \cite{AMP10} studied the problem of computing a large number of ``thick paths'' for multiple speed-bounded agents, from a source region to a sink region, where the thickness of a path models the separation standard between agents, and the objectives are to obey speed bounds, maintain separation, and maximize throughput. In the Aircraft/Train Scheduling problem (see \cite{DFKMS19,Scheffer20,Scheffer20-2}), given a set of paths on which the agents travel, and a separation parameter, the goal is to find a collision-free motion of the agents while minimizing the time of completion. (In our problems, we are not maximizing a ``throughput'' or makespan; rather, we maximize a separation standard, for a given set of agents.)

In the {\em maximum dispersion} problem, the goal is to place $n$ (static) points within a domain $P$ in order to maximize the minimum distance between two points. (Optionally, one may also seek to keep points away from the boundary of $P$.) An optimal solution provides maximum social distancing for a set of {\em static} agents, who stand at the points, without moving. Constant factor approximation algorithms are known~\cite{BF01,FM03}. (The problem is also closely related to geometric packing problems, which is a subfield in itself.)  In robotics, the problem of motion planning in order to achieve well dispersed agents has also been studied: move a swarm of robots, through ``doorways'', into a geometric domain, in order to achieve a set of agents well dispersed throughout the domain. Such movements can be accomplished~\cite{HABFM02} using local strategies that are provably competitive. 

In the adversarial setting, in which one or more agents is attempting to move in order to avoid (evade) a pursuer, there is considerable work on pursuit-evasion in geometric domains (e.g., the ``lion and man'' problem); see the survey~\cite{CHI11}.

\vspace{5pt}
\noindent\textbf{Our results.} In this paper, we (mostly) consider the case of $k=2$, i.e., two agents, ``Red'' and ``Blue'', that move inside their given domains. Further, in this paper, unless stated otherwise, we do not consider speed to be a limiting factor; e.g., when Blue moves in order to maintain distance from Red, we assume that Blue can move at a sufficient speed. 

We begin by considering the scenario in which the two domains are polygonal curves $R$ and $B$. The agents' missions are to traverse their respective curves, from the start point to the end point, in order to maximize the minimum distance between the agents. The Flipped \frechet\ measure between the two curves is the maximum separation that can be maintained.
In \Cref{sec:curves} we consider both the continuous case (agents move continuously along the edges of their curves), and the discrete case (agents ``jump'' between consecutive vertices of their curves). We first show that the Flipped \frechet\ measure between two $n$-vertex curves in one dimension (1D) can be computed in near-linear time. This is in sharp contrast with continuous \frechet\ distance, which has quadratic conditional (SETH-based) lower bounds~\cite{BM16}. We then develop quadratic or near-quadratic time algorithms for computation of discrete Flipped \frechet\ measure in 1D and 2D, and for 2D continuous Flipped \frechet\ measure. We also complement our quadratic-time algorithms with conditional lower bounds (conditioned on the Orthogonal vectors (OV) problem), even for approximation:
we give a quadratic conditional lower bound on approximating SDW for curves in 2D up to a factor better than $\frac{\sqrt{5}}{2\sqrt{2}}$, and a quadratic conditional lower bound on approximating discrete SDW for curves in 1D, with a factor better than $\frac23$.  
 
We then restrict the domain for Blue to a simple polygon $P$, and measure separation using geodesic distance in $P$. In \Cref{sec:polygon} we consider several versions. In the first, the Red agent has a mission to walk along an arbitrary path inside the polygon. The Blue agent must stay as far as possible from Red, and the only restriction is to move inside $P$. This is related to the motion planning problem in which we are given a set of disjoint bodies (e.g., disks), and we seek a motion plan that allows them to remain disjoint while each moves to a destination. We give a quadratic time algorithm for this problem. We then show that under the reasonable assumption that Red moves on a geodesic, one can compute a strategy for Blue in near linear time. 

Next, we consider a simple polygon $P$ where the Red agent is  on a mission to traverse the boundary of $P$, while the Blue agent moves within $P$ (with a starting point of Blue's choice), in order to maximize the minimum Red-Blue distance. We define the {\em social distance width} (SDW for short) of a polygon $P$ to be the minimum Red-Blue distance that can be maintained throughout the movement, maximized over all possible movement strategies. We develop quadratic time algorithm to compute the SDW of a polygon, but show that when $P$ is a skinny polygon (a tree), a strategy for Blue can be computed in linear time. The discussion on this special case of trees leads us to investigating the more general case where agents are walking on a graph (which may or may not be embedded in the plane). In \cref{sec:graphs} we define the notion of Social Distance Width of two graph, and show how to compute it.

	\section{Flipped \frechet\ measure on polygonal curves}\label{sec:curves}
\vspace{-5pt}
In this section the domains of Red and Blue are two polygonal curves $R$ and $B$, respectively. We begin by giving some basic definitions; then, we describe tools that were used in classic algorithms for \frechet\ distance and the relation to the social distancing problem for curves.

A polygonal curve $P$ in $\reals^d$ is a continuous function $P:[1,n]\rightarrow \reals^d$, such that for any integer $1\le i\le n-1$ the restriction of $P$ to the interval $[i,i+1]$ forms a line segment. We call the points $P[1],P[2],\dots,P[n]$ the vertices of $P$, and say that $n$ is the length of $P$. For any real numbers $\alpha,\beta\in[1,n]$, $\alpha\le\beta$, we denote by $P[\alpha,\beta]$ the restriction of $P$ to the interval $[\alpha,\beta]$. Then, for any integer $1\le i\le n-1$, $P[i,i+1]$ is an edge of $P$.
A continuous, non-decreasing, surjective function $f:[0,1]\rightarrow [1,n]$ is called a \emph{traversal} of $P$.

Let $P:[1,n]\rightarrow \reals^d$ and $Q:[1,m]\rightarrow \reals^d$ be two polygonal curves. A \emph{traversal} of $P$ and $Q$ is a pair $\tau=(f,g)$, with $f:[0,1]\rightarrow [1,n]$ a traversal of $P$, $g:[0,1]\rightarrow [1,m]$ a traversal of $Q$.

\begin{definition}[Flipped \frechet\ Measure]
	The \emph{(Flipped \frechet\ measure} (FF) of $P$ and $Q$ is $\displaystyle FF(P,Q)=\sup_{\tau=(f,g)}\min_{t\in [0,1]}\Vert P(f(t))-Q(g(t))\Vert.$
\end{definition}

\vspace{-5pt}
Note that the well-studied \emph{\frechet\ distance} between $P$ and $Q$ is $\inf_{\tau=(f,g)}\max_{t\in [0,1]}\Vert P(f(t))-Q(g(t))\Vert$, where $\tau$ is a traversal of $P$ and $Q$.

\vspace{5pt}
\noindent\textbf{The discrete case.} 
%\paragraph{The discrete case.} 
When considering discrete polygonal curves, we simply define a polygonal curve $P$ as a sequence of $n$ points in $\reals^d$. We denote by $P[1],\dots,P[n]$ the vertices of $P$, and for any $1\le i\le j\le n$ let $P[i,j]=(P[i],P[i+1],\dots,P[j])$ be a subcurve of~$P$.

Consider two sequences of points $P,Q$ of length $n$ and $m$, respectively. % and $Q \in \reals^{d\times m}$.
A \emph{traversal} $\tau$ of $P$ and $Q$ is a sequence, $(i_1,j_1),\dots,(i_t,j_t)$, of pairs of indices such that $i_1=j_1=1$, $i_t=n$, $j_t=m$, and for any pair $(i,j)$ it holds that the following pair is $(i,j+1)$, $(i+1,j)$, or $(i+1,j+1)$.

\begin{definition}[Discrete Flipped \frechet\ Measure]		
	The \emph{discrete Flipped \frechet\ measure} (dFF) of $P$ and $Q$ is 
	$\displaystyle dFF(P,Q)=\max_{\tau}\min_{(i,j)\in\tau}\Vert P[i]-Q[j]\Vert.$ 
\end{definition}

\vspace{-5pt}
Notice that unlike in the continuous case, the distances between the agents are only calculated at the vertices of the polygonal curves.

The \emph{discrete \frechet\ distance} (DFD) between $P$ and $Q$ is ${\min_{\tau}\max_{(i,j)\in\tau}\Vert P[i]-Q[j]\Vert}$, and it can be computed in $O(nm)$ time~\cite{EM94} using a simple dynamic programming algorithm.

From now on, we assume for simplicity that both curves $R$ and $B$ have length $n$; however, our algorithms and proofs can be easily adapted to the general case of $m\neq n$.

We give (\Cref{sec:curves-1dcts}) a near-linear algorithm to compute the continuous FF measure in 1D, demonstrating that ``flipping'' the objective function makes this setting easier: for continuous \frechet\ there exist conditional quadratic lower bounds~\cite{BM16}. We give quadratic algorithms and then conditional lower bounds (\Cref{sec:curves-quadalgos,sec:curves-quadlbs}) for computing or approximating other variants (1D discrete, 2D continuous and discrete) of FF measure, specifically:
\begin{itemize}
	\item a quadratic lower bound, conditioned on the Strong Exponential Time Hypothesis (SETH), on approximating FF measure for curves in 2D, with approximation factor $\frac{\sqrt{5}}{2\sqrt{2}}$. 
	\item a quadratic lower bound, conditioned on the Strong Exponential Time Hypothesis (SETH),  on approximating dFF measure for 1D curves, with approximation factor $\frac23$.
\end{itemize}

\subsection{A near linear time algorithm for FF in 1D}\label{sec:curves-1dcts}
Consider the decision version of FF in 1D: Given two paths, $R$ and $B$, on the $x$-axis, each specified by $n$ points, we are to determine if it is possible to find a traversal of the two paths so that Red and Blue maintain separation larger than $\delta$ during their traversals. We assume, w.l.o.g., that $R[1]<B[1]$. We can also assume that $B[1]-R[1]>\delta$, as otherwise we can simply return NO. Notice that it suffices to consider the problem of maintaining separation at least 0 between $R$ and a shifted (by $\delta$) copy of $B$; we seek to determine if $FF(B,R)>0$. 
Let $b_{\max}$ be the rightmost point of $B$ and let $r_{\min}$ be the leftmost point of $R$.

\begin{lemma}\label{lem:extreme_points} 
	If $FF(B,R)>0$, then:
	\begin{enumerate}[(i)]
		\item All of $R$ must be to the left of $b_{\max}$, and all of $B$ must be to the right of $r_{\min}$.
		\item While Blue is at $b_{\max}$, any subpath of $R$ can be traversed by Red. While Red is at $r_{\min}$, any subpath of $B$ can be traversed by Blue.
		\item There exists a traversal achieving $FF(B,R)>0$ such that at some point Blue is at $b_{\max}$ and Red is at $r_{\min}$.
	\end{enumerate}
\end{lemma}

\begin{proof}
	(i) holds by continuity: if $R$ has a point to the right of $b_{\max}$, Red must cross Blue before getting to that point. The claim for $B$ is symmetric.  (ii) then follows from (i).
	
	For (iii), consider a traversal $\tau$, and assume that Blue reaches $b_{\max}$ before Red reached $r_{\min}$. Let $r'$ be the location of Red when Blue is at $b_{\max}$; $r'$ does not have to a be a vertex of $R$, but in any case $r'$ precedes $r_{\min}$ along $R$.
	Let $b'\in B$ be Blue's location at the time Red reaches $r_{\min}$ ($b'$ is after $b_{\max}$). By (ii), while Blue is at $b_{\max}$, Red can go from $r'$ to $r_{\min}$ -- so $(b_{\max},r_{\min})$ becomes part of the traversal. Again by (ii), while Red is at $r_{\min}$, Blue can go from $b_{\max}$ to $b'$. From $(b',r_{\min})$, Blue and Red can follow $\tau$ to complete the traversal. 
\end{proof}

\Cref{lem:extreme_points} allows us to assume, w.l.o.g., that $b_{\max}$ and $r_{\min}$ are the first points of $B$ and $R$, respectively  (i.e., $B[1]=b_{\max},R[1]=r_{\min}$): for arbitrary $B,R$ we can separately solve the problem for the subpaths of Blue from $b_{\max}$ to $B[n]$ and Red from $r_{\min}$ to $R[n]$, and the problem for the (reversed) subpaths of Blue from $b_{\max}$ to $B[1]$ and Red from $r_{\min}$ to~$R[1]$.

\begin{figure}[h!]
	\centering
	\includegraphics[width=0.6\textwidth]{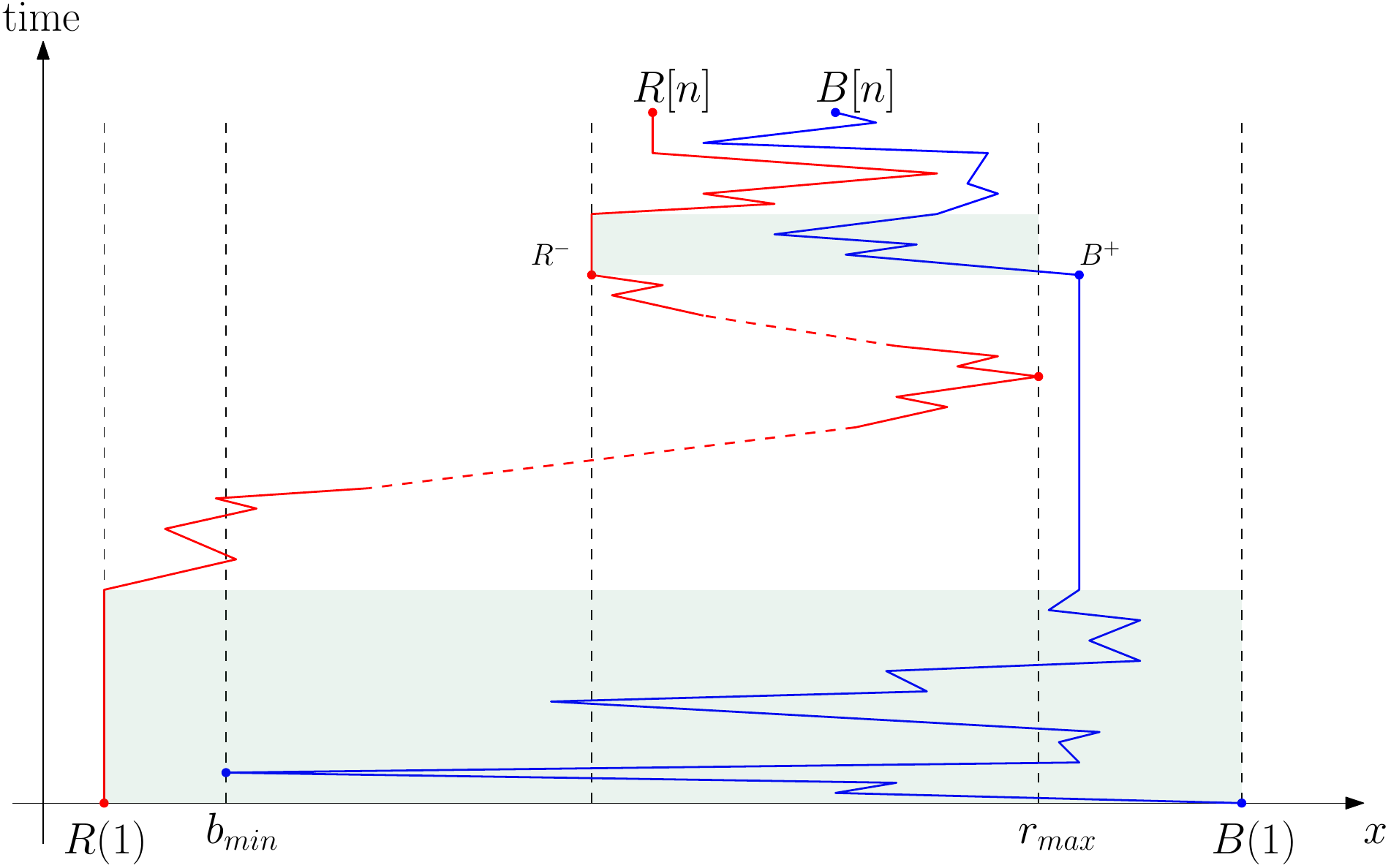}
	\vspace{-5pt}
	\caption{Determining if $SDW(B,R)>0$: viewing the traversals of the Red and Blue paths, $R$ and $B$, in space-time.}
	\label{fig:1d}
\end{figure}

Let $b_{\min}$ (resp., $r_{\max}$) be the leftmost (resp., rightmost) $x$-coordinate of $B$ (resp., $R$).
If $b_{min}\le R[1]$, or $r_{max}\ge B[1]$, then by \Cref{lem:extreme_points} (i) we get $FF(B,R)=0$. In addition, if $b_{\min}> r_{\max}$, then we see that $FF(B,R)>0$, since the $x$-coordinates of the paths do not overlap.
% \joe{something to reword here? If $b_{\min}=r_{\max}$ then the $x$-coordinates overlap at this point, but we can still maintain strictly positive distance....} \omrit{hmm I guess this should be $b_{\min}> r_{\max}$ (I changed it now)}
We are thus left with the case in which $R[1]<b_{\min}\le r_{\max} < B[1]$ (see \Cref{fig:1d}).
If $B[n]> r_{\max}$, then a feasible traversal has Red remain at $R[1]$ while Blue traverses $B$ from $B[1]$ to $B[n]$, after which Red can traverse $R$ from $R[1]$ to $R[n]$. Thus, assume that $B[n]\le r_{\max}$; this is the case shown in the figure. 

Let $B^+$ be the last vertex on the Blue path such that $B^+>r_{\max}$. Notice that such a vertex always exists because $B[1]>r_{\max}$. Let $R^-$ be the leftmost vertex on the Red path after it visits $r_{\max}$ for the last time.
\begin{lemma}
	If $FF(R,B)>0$, then there exist a traversal achieving $FF(R,B)>0$ such that at some point Blue is at $B^+$ and Red is at $R^-$.
\end{lemma}
\begin{proof}
	Consider a traversal $\tau'$ where Blue traverses the path from $B[1]$ to $B^+$ while Red is at $R[1]$, then Red traverses the path from $R[1]$ to $R^-$ while Blue is at $B^+$. Such a traversal is always possible because $B^+$ is to the right of $r_{\max}$. Now consider a traversal $\tau$ achieving $FF(R,B)>0$. We show how to complete $\tau'$ to a traversal achieving $FF(R,B)>0$.
	
	If in $\tau$ Red visits $R^-$ before Blue visits $B^+$, then let $r_1$ be the location of Red when Blue is at $B^+$. Since $B^+$ is to the right of $r_{\max}$, in $\tau'$ we can have Red traverse the path from $R^-$ to $r_1$ while Blue is at $B^+$. Then, continue as in $\tau$ to complete the traversal.
	
	Else, in $\tau$ Blue visits $B^+$ before Red visits $R^-$. Let $b_1$ be the location of Blue when Red is at $R^-$, and let $b_2$ be the location of Blue when Red is at $r_{\max}$ in the traversal $\tau$. In $\tau$, Blue traverses the subpath from $b_2$ to $b_1$ while Red traverses the subpath from $r_{\max}$ to $R^-$. Now, since $R^-$ is the leftmost point on the subpath from $r_{\max}$ to $R^-$, Blue can traverse the subpath from $b_2$ to $b_1$ while Red is at $R^-$. Notice that $b_2$ either precedes $B^+$, or it is on the edge incident to $B^+$. Therefore, in $\tau'$ we can have Blue traverse the subpath from $B^+$ to $b_1$ while Red is at $R^-$. Again we continue as in $\tau$ to complete the traversal.
\end{proof}

By the above lemma, we can simply use the following traversal: Blue traverses $B$ from $B[1]$ to $B^+$ while Red is at $R[1]$. Then Red traverses $R$ from $R[1]$ to $R^-$ while Blue is at $B^+$. The situation now is:
Blue is at $B^+$ and the subpath of Blue starting at $B^+$ is entirely to the left of $B^+$ (because $B[n]\le r_{\max}$); Red is at $R^-$ and the subpath of Red starting at $R^-$ is entirely to the right of $R^-$. Also note that $R^-$ is not the first vertex of $R$. We are thus back to our initial conditions, and can apply the algorithm recursively. 
More formally, consider the following recursive algorithm. Given two curves $R,B$ such that $R[1]\le B[1]$, $r_{\min}=R[1]$ is the leftmost point in $R$ and $b_{\max}=B[1]$ is the rightmost point in $B$:
\begin{enumerate}
	\item Find $b_{\min}$ and $r_{\max}$.
	\item If $b_{\min}\le R[1]$ or $r_{\max}\ge B[1]$, return NO.
	\item If $b_{\min}> r_{\max}$ or $B[n]>r_{\max}$, return YES.
	\item Else, find $B^+$ and $R^-$. Recurse with $B^+,\dots, B[n]$ and $R^-,\dots,R[n]$.
\end{enumerate}

\vspace{5pt}
\noindent\textbf{Running time.} 
%\paragraph{Running time.} 
We begin with a preprocessing step, where we sort the vertices of $R$ and $B$. For the decision algorithm, given a value $\delta$, we merge the sorted lists of $R$ and $B-\delta$ in linear time. Then, we can compute in linear time: 
(i) The leftmost and rightmost vertices that follow each vertex of the curve (this can be done in a single scan of the curve). (ii) For each vertex $r$ of $R$, the last vertex of $B$ that is located to the right of $r$. In each recursive step of the algorithm, we can compute the points $b_{\min}$, $r_{\max}$, $B^+$, and $R^-$, in linear time. Thus, the total running time for the decision procedure is $O(n)$.
For the optimization (finding the maximum $\delta$ that allows separation), we note that the maximum separation is achieved when Red and Blue are on vertices, because otherwise, their distance can increase by having one of them wait on a vertex. Thus, for finding the maximum possible $\delta$, we can do a binary search among the $O(n^2)$ different distances, using distance selection at each step in $O(n\log n)$ time. The distance selection algorithm stores all $O(n^2)$ distances implicitly in a sorted matrix, and then applies the sorted matrix selection algorithm of~\cite{FJ84} in $O(n\log n)$ time.

\begin{theorem}
	Given two polygonal curves $P,Q$ of length $n$ in 1D, their social distance width, $FF(P,Q)$, can be computed in $O(n\log^2 n)$ time.
\end{theorem}

\subsection{Quadratic time algorithms}\label{sec:curves-quadalgos}
In this section we describe algorithms for computing continuous and discrete $FF(R,B)$; the algorithms are based on similar algorithms for computing the continuous and discrete \frechet\ distances, in times $O(n^2
\log n)$ and $O(n^2)$, respectively. 
As with \frechet\ distance, our algorithms use the notion of a free-space diagram, appropriately adapted.

\vspace{5pt}
\noindent\textbf{The free space diagram.} 
The \emph{$\delta$-free space diagram}~\cite{AG95} of two curves $P$ and $Q$ represents all possible traversals of $P$ and $Q$ with \frechet\ distance at most $\delta$. We adapt this notion to our new setting.

Let $\C_{ij}=[i,i+1]\times[j,j+1]$ be a unit square in the plane, for integers $1\le i\le n-1$ and $1\le j\le m-1$. Let $\B=[1,n]\times[1,m]$ be the square in the plane that is the union of the squares $C_{ij}$. Given $\delta>0$, the \emph{$\delta$-free space} is $\F_\delta=\{ (p,q)\in \B \mid \Vert P(p)-Q(q)\Vert \ge \delta \}$. In other words, it is the set of all red-blue positions for which the distance between the agents is at least $\delta$. A point $(p,q)\in \F_\delta$ is a {\em free point}, and the set of {\em non-free points} (or {\em forbidden points}) is then $\B\setminus\F_\delta$. Note that for \frechet\ distance, these definitions are reversed (``flipped'').
We call the squares $\C_{ij}$ the \emph{cells} of the free space diagram; each cell may contain both free and forbidden points. An important property of the free space diagram is that the set of forbidden points inside a cell $\C_{ij}$ (i.e., $\C_{ij}\cap\F_\delta$) is convex~\cite{AG95}.

\begin{figure}[h!]
	\centering
	\includegraphics[scale=0.8,page=1]{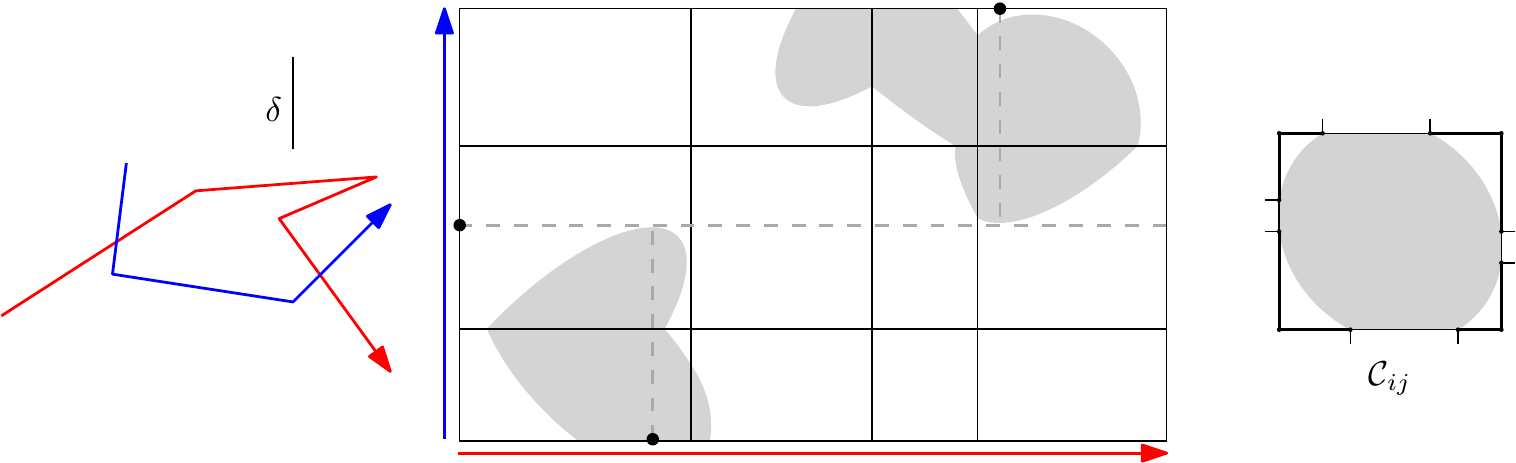}
	\caption{Right: a free space cell $\C_{ij}$. For FF, the free space is white, while for \frechet\ distance, the free-space is gray. The gray region within a cell is convex. Left: the free space diagram of two curves. The black points and dashed lines indicate critical values of type (iii), which are openings in the free space diagram defined by two red edges and one blue edge.}
	\label{fig:free-space}
\end{figure}
Notice that a monotone path through the free space $\F_\delta$ between two free points $(p,q)$ and $(p',q')$ corresponds to a traversal of $P[p,p']$ and $Q[q,q']$. Thus, $FF(P,Q)\ge \delta$ if and only if there exists a monotone path through the free space $\F_\delta$ between $(0,0)$ and $(n,n)$ (i.e. $(n,n)$ is ``reachable'' from $(0,0)$). The \frechet\ distance between $P$ and $Q$ can be computed in $O(n^2\log n)$ time~\cite{AG95} as follows. For a given value of $\delta$, the reachability diagram is defined to be the set of points in $\F_\delta$ reachable from $(0,0)$. As the set of free points in each cell is convex, the set of ``reachable'' points on each of the boundary edges of a cell is a line segment. Thus, one can construct in constant time the reachable boundary points of a cell, given the reachable boundary points of its bottom and left neighbor cells (see \Cref{fig:free-space}). Therefore, computing the reachability diagram (and hence solving the decision version of the problem) takes $O(n^2)$ time using a dynamic programming algorithm.
For the optimization, there are $O(n^3)$ critical values of $\delta$, which are defined by (i) the distances between starting points and endpoints of the curves, (ii) the distances between vertices of one curve and edges of the other, and (iii) the common distance of two vertices of one curve  to the intersection point of the bisector with some edge of the other. Then, parametric search, based on sorting, can be performed in time $O((n^2+T_{dec})\log n)$, where $T_{dec}$ is the running time for the decision algorithm.

In the case of $FF$, we can again compute the reachability diagram in $O(n^2)$ time, as in each cell the set of forbidden points is convex, and thus the set of ``reachable'' points on each of the boundary edges of a cell is at most two line segments. The set of critical values is similar, except that the third type can occur between three edges (see \Cref{fig:free-space}, left). Thus, by arguments similar to~\cite{AG95}, we have a $O(n^2\log n)$ time algorithm for computing $FF(P,Q)$.

\begin{restatable}{theorem}{contFF}
	There is an $O(dn^2 \log n)$ time exact algorithm for computing the Flipped \frechet\ measure of two polygonal $n$-vertex curves in $\reals^d$. 
\end{restatable}

For the discrete version of $FF$, a simple dynamic programming algorithm (similar to the one known for discrete \frechet\ distance) gives an $O(n^2)$ solution. In short, let $OPT[i,j]$ be the FF of $P[1,i]$ and $Q[1,j]$; then, by the definition of $dFF$ we have $OPT[i,j]=\min\{\|P[i]-Q[j]\|,\max\{OPT[i-1,j],OPT[i,j-1],OPT[i-1,j-1]\}\}$.

\begin{restatable}{theorem}{discrteFF}
	There exists an $O(dn^2)$ time exact algorithm for computing the $dFF$ measure of two polygonal $n$-vertex curves in $\reals^d$. 
\end{restatable}

\subsection{Quadratic lower bounds}\label{sec:curves-quadlbs}
Bringmann and Mulzer \cite{BM16} give a lower bound (conditioned on SETH) for computing the discrete \frechet\ distance between two curves in 1D, using a reduction from the Orthogonal Vectors (OV) problem. We prove similar results below for FF in 2D and dFF in 1D.

The Orthogonal Vectors problem is defined as follows. Given two sets $U=\{u_1\ldots u_N\}$ and $V=\{v_1\ldots v_N\}$, each consisting of $N$ vectors in $\{0,1\}^D$, decide whether there are $u_i\in U,v_j\in V$ orthogonal to each other, i.e., $u_i(k)\cdot v_j(k)=0$ for every $k=1,\dots, D$ (where $u_i(k)$ denotes the $k$th coordinate of $u_i$). Bringmann and Mulzer~\cite{BM16} showed that if OV has an algorithm with running time $D^{O(1)}\cdot N^{2-\eps}$ for some $\eps>0$, then SETH fails.

In the following, by an algorithm with approximation factor $\alpha<1$ we mean an algorithm that outputs a traversal whose maintained separation distance is \textit{at least} $\alpha$ times the FF, given by an optimal traversal.

\subsubsection{Flipped \frechet\ in 2D}

We first show that for continuous FF, increasing the dimension from 1D (where we gave a near-linear-time algorithm) to 2D likely rules out subquadratic algorithms. 

\vspace{-4pt}
\begin{theorem}\label{thm:ctslowerbound}[FF Lower Bound in 2D] There is no algorithm that computes the FF measure between two polygonal curves of length $n$ in the plane, up to an approximation factor at least $\frac{\sqrt5}{2\sqrt{2}}$, and runs in time $O(n^{2-\delta})$ time for any $\delta>0$, unless OV fails.
\end{theorem}

\begin{proof}
	Set $\alpha=\frac{\sqrt5}{2\sqrt{2}}$. Given an instance $U=\{u_1\ldots u_N\}$, $V=\{v_1\ldots v_N\}$ of OV, we show how to construct two curves $R$ and $B$ of length $O(N\cdot D)$ in the plane, such that if $U\times V$ contains an orthogonal pair then $FF(R,B)\ge1$, and if $U\times V$ does not contain an orthogonal pair then $FF(R,B)\le\alpha$. Therefore, if there exists an algorithm with running time $O(n^{2-\eps})$ for computing $FF$ of two curves of length $n$ in the plane, then $FF(R,B)$ can be computed in $O((DN)^{2-\eps})$, which means that OV can be decided in $D^{O(1)}\cdot N^{2-\eps}$ time, and SETH fails. Moreover, if $FF$ can be approximated up to a factor of $\alpha$, then again we get that OV can be decided in $D^{O(1)}\cdot N^{2-\eps}$ time, and SETH fails. As in \cite{BM16}, we assume that $D$ is even (otherwise add a 0 coordinate to each vector).
	
	The construction of $R$ (resp. $B$) is such that for each vector $u_i\in U$ (resp. $v_j\in V$), we construct a vector gadget curve $A_i$ (resp. $B_j$), such that if $u_i$ and $v_j$ are orthogonal then $FF(A_i,B_j)=1$, and otherwise $FF(A_i,B_j)\le\alpha$. Then, we connect the vector gadgets into curves $R$ and $B$.
	
	Consider the following set of points (see \Cref{fig:ctscurves}):
	\begin{itemize}
		\item Points on the $x$-axis: $x=(-1.5,0)$, $r=(-0.75,0)$, $r'=(-0.5,0)$, $b=(0.25,0)$, $b'=(0.5,0)$
		\item Points on the $y$-axis: $c^u_0=-c^d_0=(0,0.75)$, $c^u_1=-c^d_1=(0,0.25)$, $u=-d=(0,0.5)$
		\item Octagon points: $p_1=-p_5=(0.25,0.5)$, $p_2=-p_6=(0.5,0.25)$, $p_3=-p_7=(0.5,-0.25)$, $p_4=-p_8=(0.25,-0.5)$
		\item Other points: $s_1=(1,0.5)$, $t_1=(1,-0.5)$, $s_2=(-0.25,-0.25)$, $t_2=(0.25,0.25)$
	\end{itemize}
	Notice that all the points are located on a regular grid with side length $0.25$, and that $\| s_1-t_2\|=\| t_1-s_2\|=\| c^u_1-p_4\|=\| c^u_1-p_5\|=\| c^d_1-p_1\|=\| c^d_1-p_8\|=\|r-c^u_1\|=\|r-c^d_1\|=\alpha$.
	
	\begin{figure}[h]
		\centering
		\includegraphics[trim = 0cm 3cm 3cm 3cm, clip, page=2, scale=0.7]{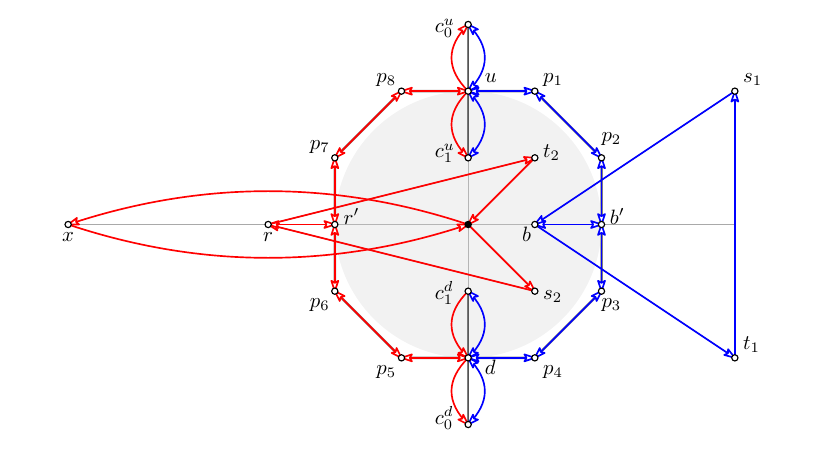}
		\vspace{-5pt}
		\caption{\small The curves $R$ and $B$ for the continuous FF lower bound.}\label{fig:ctscurves}
	\end{figure}
	
	For each $u_i\in U$, the gadget $A_i$ is constructed as follows: 
	\[r' \bigcirc_{k=1,\dots,\frac{D}{2}} \left(p_7 \circ p_8 \circ u \circ c^u_{u_i(2k-1)} \circ u \circ p_8 \circ p_7 \circ p_6 \circ p_5 \circ d \circ c^d_{u_i(2k)} \circ d \circ p_5 \circ p_6\right) \circ r'~.\]
	Similarly, for $v_j\in V$, the gadget $B_i$ is constructed as follows: 
	\[b' \bigcirc_{k=1,\dots,\frac{D}{2}} \left(p_3 \circ p_4 \circ d \circ c^d_{v_i(2k-1)} \circ d \circ p_4 \circ p_3 \circ p_2 \circ p_1 \circ u \circ c^u_{v_i(2k)} \circ u \circ p_1 \circ p_2\right) \circ b'~.\]
	It is easy to see that if $u_i,v_j$ are orthogonal then $FF(A_i,B_j)=1$, because the traversal that uses ``antipodal'' points maintains distance 1 between Red and Blue.
	For the other direction, we claim that if $FF(A_i,B_j)>\alpha$, then $u_i,v_j$ are orthogonal.
	Now assume that $u_i(1)=1$, so $R$ starts with $r',p_7,p_8,u,c^u_1$.
	Notice that when Blue traverses the subcurve $b',p_3,p_4,d$, Red cannot reach $c^u_1$. If $v_j(1)=1$, then the next move of Blue is toward $c^d_1$, and the distance between the agents becomes at most $\alpha$ (if Red in on $p_8$ while Blue is on $c^d_1$ then their distance is exactly $\alpha$). This means that $u_i(1)=v_j(1)=1$ is not possible. 
	Therefore, at least one of $u_i(1),v_j(1)$ is $0$. Notice that Blue will visit $d$ for the first time before Red visits $d$ for the first time. Similarly, Red will visit $u$ for the first time before Blue visits $u$ for the first time. Moreover, Red cannot reach $r'$ before Blue leaves $d$, and Blue cannot reach $b'$ before Red leaves $u$. Thus, the movement of Red and Blue is synchronized in the sense that when Red is moving from $u$ towards $d$, Blue is moving from $d$ towards $u$, and vice versa. Hence, by similar (symmetric) arguments, we get that $u_i(k)=v_j(k)=1$ is not possible for all $k=2,\dots,D$ as well, and $u_i,v_j$ are orthogonal.

	The gadgets $A_i$ and $B_j$ are connected into $R$ and $B$ as follows: \[R=x\circ (0,0)\circ s_2 \bigcirc_{i=1,\dots,N} \left( r\circ A_i \right) \circ r \circ t_2\circ(0,0)\circ x~, \qquad
	B=\bigcirc_{j=1,\dots,N} \left( s_1\circ b \circ B_j \circ b \circ t_1 \right)~.\]
	If $u_i\in U,v_j\in V$ are orthogonal, then Blue traverses 
	$\bigcirc_{k=1,\dots,j-1} \left( s_1\circ b \circ B_k \circ b \circ t_1 \right)$ while Red stays at $x$, then Blue moves to $s_1$. Now Red moves from $x$ to $(0,0)$, to $s_2$, then traverses $\bigcirc_{k=1,\dots,i-1} \left( r\circ A_k \right)$, and moves to $r$ just before $A_i$. Now Blue moves to $b$ just before $B_j$, and they traverse $A_i$ and $B_j$ in sync, keeping distance $\geq 1$. Now Red moves to $r$ while Blue moves to $b$, then Blue moves to $t_1$. While Blue is on $t_1$, Red traverses $\bigcirc_{k=i+1,\dots,N} \left( r\circ A_k \right) \circ r' \circ t_2 \circ x$. Finally, Blue traverses $\bigcirc_{k=j+1,\dots,N} \left( s_1\circ b \circ B_k \circ b \circ t_1 \right)$ while Red is on $x$.
	
	For the converse, assume that $FF(R,B)>\alpha$.
	When Red reaches $(0,0)$ for the first time, Blue must be on $s_1$ or $t_1$ (or on the edge between them). If Blue is on $t_1$, then it must move towards $s_1$ before Red can continue to $s_2$; moreover, when Red reaches $s_2$ Blue can only be strictly above the $x$-axis either on the edge $(t_1,s_1)$ or $(s_1,b)$. Thus we can assume that when Red is on $s_2$, Blue is on $s_1$ immediately before some vector gadget $B_j$. Now consider the first time when Red reaches $t_2$. At this time, Blue can be either near $t_1$ (strictly below the $x$-axis) or near $c^d_0$. However, if Blue is near $c^d_0$, it is not possible for the agents to continue their movement: Blue cannot reach $d$ and Red cannot reach $(0,0)$. Thus we can assume that when Red reaches $t_2$, Blue is near $t_1$ and strictly below the $x$ axis. We conclude that between the first time that Red visited $s_2$ and the first time that Red visited $t_2$, Blue had to traverse the vector gadget $B_j$ in order to get from a point strictly above the $x$ axis to a point strictly below it. Before Blue starts traversing $B_j$, it first visits $b$, but when Blue is on $b$, the only possible location of Red is near $r$ (more precisely, on the edge between $r$ and $r'$, but not on $r'$). Notice that Red cannot be near $x$ since it has not visited $t_2$ yet. Now there are two options: (1) Red is immediately before some vector gadget $A_i$, and thus by previous arguments $u_i$ and $v_j$ are orthogonal, or (2) Red has finished all its vector gadgets, and it waits on $r$ while Blue traverses $B_j$. In this case, notice that since $\|r-c_1^u\|= \|r-c_1^d\|=\alpha$, Blue cannot reach any of $c_1^u,c_1^d$ while Red is on $r$; thus, $v_j$ must be a $0$-vector, implying $u_i$ and $v_j$ are orthogonal.
\end{proof}

\subsubsection{Discrete flipped \frechet\ in 1D}

\begin{theorem}\label{1dlowerbound}[Discrete 1D Lower Bound]
	There is no $O(n^{2-\eps})$ time $\alpha$-approximation algorithm for dFF in 1D, for any $\eps>0$ and $\alpha>2/3$, unless OV fails.
\end{theorem}

\begin{proof}
	Given an instance $U=\{u_1\ldots u_N\}$, $V=\{v_1\ldots v_N\}$ of OV, we show how to construct two curves $R$ and $B$ of length $O(N\cdot D)$ on the line, such that if $U\times V$ contains an orthogonal pair then $dFF(R,B)\ge1$, and if $U\times V$ does not contain an orthogonal pair then $dFF(R,B)\le2/3$. Therefore, if there exists an algorithm with running time $O(n^{2-\eps})$ for computing $dFF$ of two curves of length $n$ on the real line, then $dFF(R,B)$ can be computed in $O((DN)^{2-\eps})$, which means that OV can be decided in $D^{O(1)}\cdot N^{2-\eps}$ time, and SETH fails. Moreover, if $dFF$ can be approximated up to a factor of $2/3$, then again we get that OV can be decided in $D^{O(1)}\cdot N^{2-\eps}$ time, and SETH fails. Again as in \cite{BM16} we assume that $D$ is even.

	\begin{figure}[h!]
		\centering
		\includegraphics[scale=0.8]{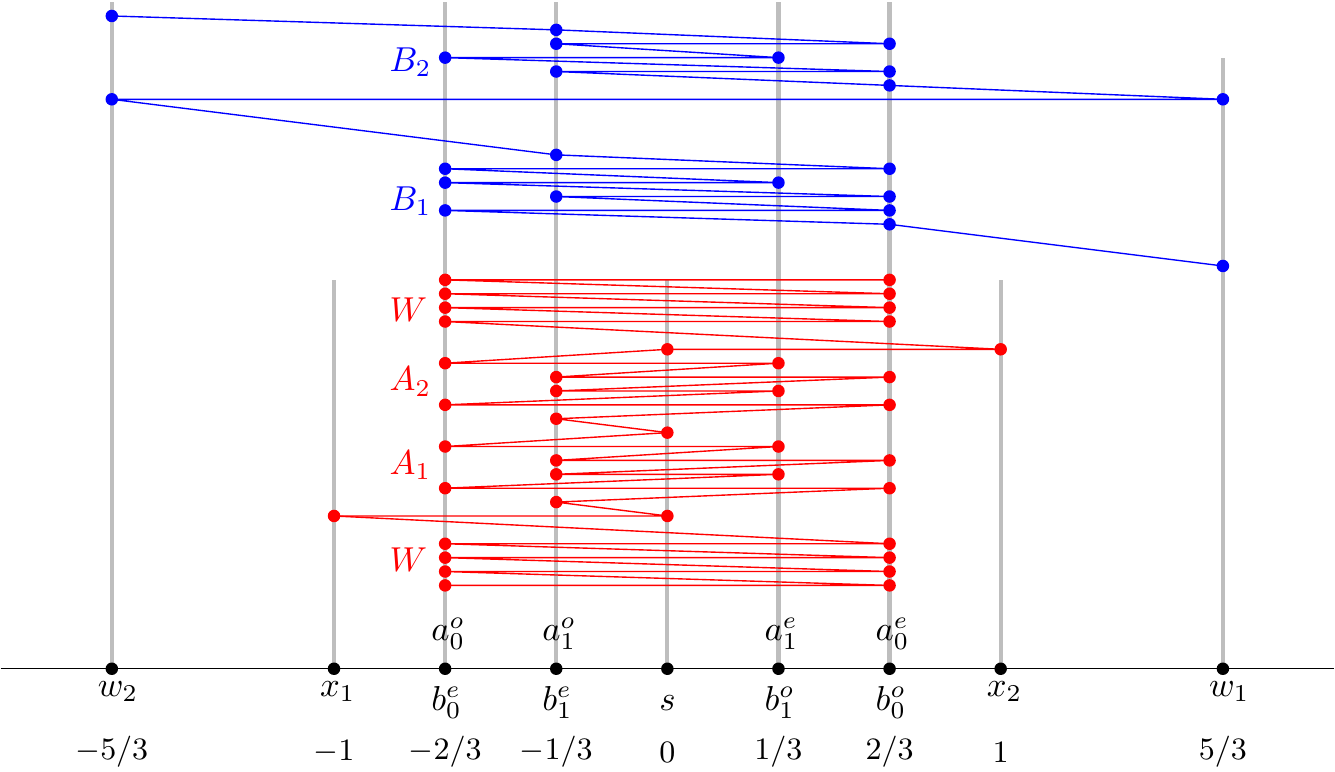}
		\caption{For a pair of orthogonal vectors, Red and Blue jump in sync between $a$s and $b$s, with at least one of them at a "far" point (indexed with 0)}\label{fig:DFF1d}
	\end{figure}
	
	Consider the following set of points on the line (see~\Cref{fig:DFF1d}): $w_1=-w_2=5/3$, $x_2=-x_1=1$, $a^e_0=b^o_0=-a^o_0=-b^e_0=2/3$, $a^e_1=b^o_1=-a^0_1=-b^e_1=1/3$, $s=0$.

	% For the reduction, 
	We first construct vector gadgets. 
	For each $u_i\in U$, we create a subsequence $A_i$ of $R$: for odd (resp. even) $k$, the $k$th point in $A_i$ is $a^o_{u_i(k)}$ (resp. $a^e_{u_i(k)}$). Similarly, for $v_j\in V$, we create a subsequence $B_j$ of $B$, using $b$s instead of $a$s. It is easy to see that Red and Blue can traverse $B_j$ and $A_i$ while maintaining distance 1 if and only if $u_i,v_j$ are orthogonal (they jump between odd and even points in sync, "opposite" each other, and at least one of them is at "far" point, indexed with 0). Note that any such vector gadget has length $D$.
	
	The gadgets $B_j$ are connected into $B$ as follows: 
	$B=w_1\circ B_1\circ w_2\circ w_1\circ B_2\circ w_2\circ\cdots\circ w_1\circ B_N\circ w_2$. Let $W$ be the sequence of $D(N-1)$ points that alternate between $a^o_0$ and $a^e_0$ starting with $a^o_0$ (Red can traverse each of $N-1$ length-$D$ subpaths of $W$ in sync with Blue on any $B_j$).
	We construct $R$ as follows:
	$R=W\circ x_1\circ s\circ A_1\circ s\circ A_2\circ\cdots\circ s\circ A_N\circ s\circ x_2\circ W$.
	
	The proof that OV is a Yes instance if and only if $dFF(R,B)\ge1$ is similar to that in \cite{BM16}, with an important change: here, we also show that if OV is a No instance, then $dFF(R,B)\le2/3$. Indeed, notice that since $dFF(R,B)$ is determined by the distance between two vertices, one from $R$ and one from $B$, we get that if $dFF(R,B)>2/3$, then necessarily $dFF(R,B)\ge 1$ (as there are no two points in our construction with distance in the range $(2/3,1)$). Therefore, its is enough to show that if OV is a No instance, then $dFF(R,B)< 1$.
	
	If $u_i\in U,v_j\in V$ are orthogonal, then Red traverses $D(N-j)$ points on $W$ while Blue stays at $w_1$, then Blue traverses $B_1\ldots B_{j-1}$ in sync with Red traversing the rest of $W$.
	Now, while Blue stays at $w_1$ before $B_j$, Red goes to $x_1$ and traverses $A_1\ldots A_{i-1}$, then goes to $s$ before $A_i$. Then, $A_i,B_j$ are traversed in sync, Blue stays at $w_2$ while Red completes the traversal of $A_{i+1},\dots,A_N$ and goes to $x_2$, and finally Blue can complete the traversal of $B_{j+1},\dots,B_N$ in sync with Red traversing the second $W$ gadget. When Blue goes to $w_2$, Red is able to complete the traversal of $W$.
	
	For the converse, assume that $dFF(R,B)\ge 1$.
	When Red is on $x_1$, Blue must be to the right of $s$, but if Blue is not on $w_1$, then they cannot take the next step -- so Blue must be on $w_1$, right before some vector gadget $B_j$. 
	Immediately after leaving $w_1$, Blue gets to $b^o_0$ or $b^o_1$, implying that Red must be at $a^o_0$ or $a^o_1$, i.e., either in some vector gadget $A_i$ or on the second $W$ (since it already passed $x_1$). However, if it is on $W$, it must have gone through $x_2$ which is too close to $w_1$, so Red is in $A_i$. Now while they are on $A_i$ and $B_j$, Red and Blue must jump in sync, until one of them reaches the end point of their respective vector gadget. Therefore, if Red did not start on the first point of $A_i$, then it will finish $A_i$ and appear at $s$ before Blue has finished $B_j$. This means that $A_i$ and $B_j$ where traversed simultaneously, and since $dFF(R,B)\ge 1$, the respective vectors have to be orthogonal.
\end{proof}

\subsection{More than 2 agents}
Our algorithm for continuous FF in 1D generalizes to any number $k\ge2$ of agents, with a running time of $O(kn\log n)$, as follows. Let $A_1, A_2,\dots,A_k$ be the polygonal paths of $k$ agents in 1D. Given a distance value $\delta$, our goal is to decide whether the $k$ agents can traverse their respective paths while maintaining distance $\delta$ from one another. Since we are on a line and agents cannot cross paths, we only need to maintain distances for neighboring (along the line) agents. Therefore, we can translate each $A_i$ by $-(i-1)\delta$ (for $2\le i\le k$), and then our goal is to find non-crossing traversals, or corresponding paths on the space-time graph. This can be done by fixing a path for $A_1$, then using our algorithm for two agents to align a corresponding path for $A_2$, then consider the path of $A_2$ as fixed, and align a path for $A_3$ and so on (recall that we can have an agent walk in infinite speed).

The question of whether or not there exists an algorithm in 2D with running time fully polynomial in $k$ remains open (for the \frechet\ distance of a set of curves, the best known running time is roughly $O(n^k)$; see \cite{DR04}).

	\section{Social distancing in a simple polygon}\label{sec:polygon}
\vspace{-5pt}
In this section we consider distancing problems in which the given domain (for both Red and Blue) is a simple polygon. Since the two agents are moving inside the same polygon, it is natural to consider geodesic distance (i.e., the shortest path inside the polygon) instead of Euclidean distance to measure separation.

Consider a scenario in which Red and Blue have to traverse two polygonal paths $R$ and $B$, both inside a given polygon $P$, and their goal is to find a movement strategy (a traversal) that maintains geodesic distance of at least $\delta$ between them. For the analogous \frechet\ problem (Red and Blue have to maintain geodesic distance of at most $\delta$), Cook and Wenk~\cite{CW10} presented an algorithm that runs in $O(n^2 \log N)$ time, where $N$ is the complexity of $P$ and $n$ is the complexity of $R$ and $B$. Their algorithm is based on the fact that the free space in a cell of the diagram is $x$-monotone, $y$-monotone, and connected. Then the geodesic decision problem can be solved by propagating the reachability information through a cell in constant time, as for the Euclidean \frechet\ distance. Thus, we can apply a similar ``flipped'' algorithm for computing the $FF(R,B)$ under geodesic distance in nearly quadratic time.

When both Red and Blue are restricted to traverse a given path, it seems that the \frechet-like nature of the problem leads to near-quadratic time algorithms. Thus, in this section we consider the scenario where Blue has more freedom, and it is not required to traverse a given path. We first consider the case when Red is walking on an arbitrary path in the polygon; we show that while the naive solution takes at least cubic time, there exists a quadratic time algorithm for this problem. We then describe \textit{two variants for which we present linear time algorithms}; the first is where Red is walking along a shortest path in the polygon, and second is where Red is traversing the boundary of a skinny polygon (a tree). 
\SDcenter{Finally, we define the ``social distancing center'' of a polygonal domain, and give a polynomial-time algorithm to compute it.} 

Throughout this section, we will use SDW to denote social distancing width, to differentiate from the Flipped \frechet\ versions where both the Red and Blue curves were given as input.

\subsection{Red on an arbitrary path mission}\label{sec:polygon-arbitrary-path}
Consider the case when Red moves along an arbitrary path $R$ in $P$, and Blue may wander around in $P$, starting from some given point $b$. The free-space diagram can be adapted to the case of a path and a polygon, by partitioning the polygon into a linear number of convex cells (for example, a triangulation). This is a three-dimensional structure, which contains $O(n^2)$ cells (assuming that the complexity of both $R$ and $P$ is $O(n)$). However, for maintaining geodesic separation, building the free-space may be cumbersome, because it would involve building the parametric shortest path map in a triangle as the source moves along a segment, and such SPM may have $\Omega(n)$ combinatorial changes. Instead, we show that Blue may stay on the boundary, thus reducing the problem to the standard free-space diagram between a path and a closed curve. We will prove:

\begin{theorem}
	Let $P$ be a polygon with $n$ vertices, $b$ a point in $P$, and $R$ a path between two points $r$ and $r'$ in $P$. There exists an $O(n^2)$-time algorithm to decide whether there exists a path $B$ in $P$ starting from $b$, such that $SDW(R,B)>1$ under geodesic distance.
\end{theorem}

The proof follows from the following observation and lemma.

\begin{observation}\label{lem:geodesic-components}
    Let $p$ be an arbitrary point inside a polygon $P$, and $D$ be the closed geodesic disk of any radius centered at $p$, then $D$ splits both $P$ and $\partial P$ into the same number of connected components, and there is a natural one-to-one correspondence between them.
\end{observation}

\begin{lemma}\label{lem:Blue-on-boundary}
    Assume that the point $b$ is on the boundary. If there exists a path $B$ in $P$ starting from $b$, such that $SDW(R,B)>1$ under geodesic distance, then there exists such a path $B'$ that is entirely on the boundary of $P$.
\end{lemma}
\begin{proof}
    Let $r_1$ be the first point of $R$, and denote by $C_1,C_2,\dots,C_k$ the set of connected components of $P\setminus D_{r_1}$ (where $D_{r_1}$ is the unit geodesic disk around $r_1$ as in \Cref{sec:polygon-shortest-path}). Assume that $b$ lies in $C_1$. When Red moves along $R$, some of the connected components may \textbf{disappear}, \textbf{split}, or \textbf{merge} with other connected components, and some new connected components may \textbf{appear}. More formally, we can say that a sequence of connected components of $P\setminus D_{r}$ (where $r$ is moving continuously in $P$) belong to the same ``combinatorial'' connected component, if its intersection converges to some non-empty area of $P$. 
    By \Cref{lem:geodesic-components}, as long as none of the above changes in the combinatorial definition of $C_1$ occur, then Blue can remain at some point on the boundary of $C_1$, because as Red moves, $C_1$ contains a single connected piece of $\partial P$ on which Blue can walk.
    Consider the first time when Red reaches a point $r_2$ on $R$ such that $C_1$ either disappear, split, or merge. If $C_1$ disappears, then Blue had no way to escape. If $C_1$ splits, then by \Cref{lem:geodesic-components}, Blue can move on the boundary of $C_1$ to any of the new connected components, right before the split occurs. If $C_1$ merges with another connected component, say $C_2$, then again by \Cref{lem:geodesic-components}, Blue can move to the boundary of $C_2$ via the boundary of $C_1$. In any of this cases, Blue can move via $\partial P$ to the connected component in which it was if it would walk along $B$.
\end{proof}

\begin{wrapfigure}{R}{0.25\textwidth}
\vspace{-15pt}
	\centering\includegraphics[scale=0.6]{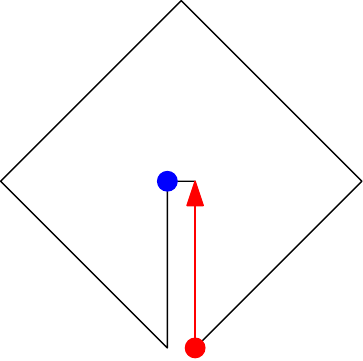}
	\end{wrapfigure}
\noindent\textbf{Remark.} 
        The above lemma applies only to geodesic distance separation - if Blue is maintaining separation using Euclidean distance, Blue may need to go inside the polygon (see the figure on the right). In this case we build the 3D free-space, but note that its complexity is quadratic and it can be searched in quadratic time.

\subsection{Red on a shortest path mission}\label{sec:polygon-shortest-path}
\vspace{-5pt}
Assume that Red moves along a geodesic path $R$ in $P$ (Red is on a mission and does not care about social distancing) while Blue may wander around anywhere within $P$ starting from a given point $b$. We show that the decision problem, whether Blue can maintain (geodesic) social distance at least 1 from Red, can be solved in linear time.

\begin{theorem}\label{thm:shortest_path}
Let $P$ be a polygon with $n$ vertices, $b$ a point in $P$, and $R$ a geodesic shortest path between two points $r$ and $r'$ in $P$. There exists an $O(n)$-time algorithm to decide whether there exists a path $B$ in $P$ starting from $b$, such that $SDW(R,B)>1$ under geodesic distance.
\end{theorem}
\vspace{-10pt}
\begin{figure}[h]
	\centering\includegraphics[scale=.6,page=1]{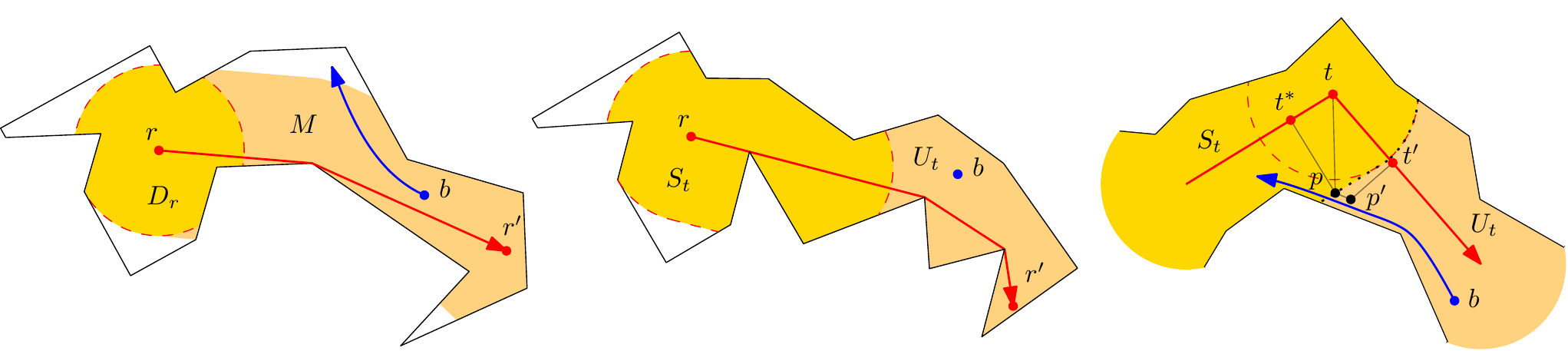}
	\caption{\small Left: $D_r$ splits $P$ into connected components. Blue escaping Red as there exists a safe point. Middle: the boundary between $S_t$ and $U_t$ is dashed. Blue cannot escape because its connected component equals $M\setminus D_r$. Right: Blue escapes when Red is at $t$.}\label{fig:mission}
	\vspace{-10pt}
\end{figure}

\noindent\textit{Proof.} 	
We use $|ab|$ to denote the geodesic distance between points $a,b\in P$. For a point $t\in\R$ let $D_t=\{p\in P:|tp|\le 1\}$ be the unit geodesic disc centered on $t$; let $M=\cup_{t\in R}D_t$ be the set of points within geodesic distance 1 from $R$ (\Cref{fig:mission}, left). 
Without loss of generality, assume $b\notin D_r$ (otherwise separation fails from the start). The disk $D_r$ splits $P$ into connected components (a component is a maximal connected subset of $P\setminus D_r$): Blue can freely move inside a component without intersecting $D_r$; in particular, if the component $P'\ni b$ of $b$ is not equal to $M\setminus D_r$ (i.e., if $P'\setminus M\ne\emptyset$), then Blue can move to a point in $P'\setminus M$ (a safe point) and maintain the social distance of 1 from Red (existence of a safe point can be determined by tracing the boundary of $M$). Next, we show that the existence of such a safe point is also \emph{necessary} for Blue to maintain the distance of 1. %In other words, if $P'=M\setminus D_r$, then there is no traversal for Blue that maintains distance at least one from Red.

Indeed, as Red follows $R$, $D_t$ sweeps $M$; let $S_t\subseteq M$ be the points swept (at least once) by the time Red is at $t\in R$ and let $U_t=M\setminus S_t$ be the unswept points. Since $b\notin D_r=S_r$, initially Blue is in the unswept region. Assume that there is no safe point ($P'= M\setminus D_r$) and yet Blue can escape. Suppose Blue can escape from the unswept to swept when Red is at $t\in R$ (\Cref{fig:mission}, right). Then there exists a point $p$ on the boundary between $U_t$ and $S_t$ that is further than 1 from $t$, say $|pt|=1+\eps$ for some $\eps>0$. Since $p$ is on the boundary of $S_t$, at some position $t^*\in\pi$ before $t$, we had $|t^*p|=1$. Since $p$ is on the boundary of $U_t$, there exists an unswept point $p'\in U_t$ within distance less than $\eps$ from $p$: $|p'p|<\eps$. Finally, since $U_t=M\setminus S_t$ is part of $M$, $p'$ becomes swept when Red is at some point $t'\in\pi$ after $t$: $|t'p'|\le 1$.
We obtain that there are three points $t^*,t,t'$ along a geodesic path $\pi$ and a point $p$ such that $|t^*p|\le1<1+\eps=|tp|$ and $|t'p|\le|t'p'|+|p'p|<1+\eps=|tp|$, contradicting the fact that the geodesic distance from a point to a geodesic path is a convex function of the point on the path \cite[Lemma~1]{PSR89} (this is the place where we use that $R$ is a geodesic path: if $R$ is not geodesic, it is not necessary for the Blue to escape from $M$ while Red is at $r$).% See also the discussion in \cref{apx:conclusion}).

\begin{wrapfigure}{R}{0.35\textwidth}
	\centering
	\centering\includegraphics[scale=0.9]{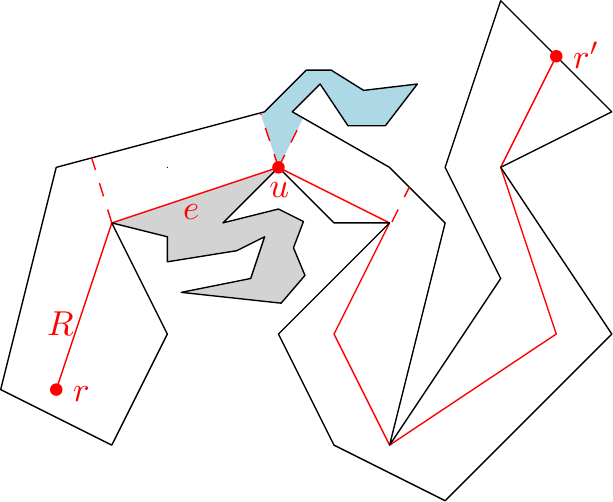}\caption{\small SPMs from edges and vertices of $\pi$ are computed separately in parts of $P$ defined by perpendiculars to path edges (some shown dashed). Gray and lightblue parts are charged to (the right side of) the edge $e$ and to vertex $u$ of $R$ resp.}\label{fig:SPM}
\end{wrapfigure}
We now show how to implement our solution to the decision problem in $O(n)$ time. To build the geodesic unit disk $D_r$ we compute the \emph{shortest path map} (SPM) from $r$ (the decomposition of $P$ into cells such that for any point $p$ inside a cell the shortest $r\textrm-p$ path has the same vertex $v$ of $P$ as the last vertex before $p$) -- the SPM can be built in linear time \cite{GH89}; then in every cell of the SPM we determine the points of $D_r$: any cell is either fully inside $D_r$, or fully outside, or the boundary of the disk in the cell is an arc of the radius-($1-|rv|$) circle centered on the vertex $v$ of $P$. The set $M$ can be constructed similarly, using SPM from $R$. To build the SPM, we decompose $P$ by drawing perpendiculars to the edges of $R$ at every vertex of the path (see \Cref{fig:SPM}): in any cell of the decomposition, the map can be built separately because the same \emph{feature} (a feature is a vertex or a side of an edge) of $R$ will be closest to points in the cell (the decomposition is essentially the Voronoi diagram of the features). In every cell, the SPM from the feature can be built in time proportional to the complexity of the cell (the linear-time funnel algorithm for SPM \cite{GH89} works to build the SPM from a segment too: the algorithm actually propagates shortest path information from segments in the polygon). Since the total complexity of all cells is linear, the SPM is built in overall linear time. After $D_r$ and $M$ are built, we test whether $b\in D_r$ (if yes, the answer is No) and trace the boundary of $M$ to determine the existence of a safe point (the answer is Yes iff such a point exists).
\qed

\subsection{SDW of closed curves and polygons}

Consider a scenario in which the polygonal curves $R$ and $B$ are closed curves. Here, the starting points of Red and Blue are not given as an input, and the goal is to decide whether they can traverse their respective curves while maintaining distance at least $\delta$. The analogous \frechet\ problem has been investigated by Alt and Godau~\cite{AG95}, who presented an $O(n^2\log^2 n)$ time algorithm, and later by Schlesinger et. al.~\cite{SVY14}, who improved the running time to $O(n^2\log n)$. Those algorithms include the construction of dynamic data structures for the free space diagram, which is again based on the fact that the free space within a cell is convex. Since, in our ``flipped'' case, the forbidden space is convex, similar data structures can be used in order to compute the SDW of two closed curves in near quadratic time.

We can then define the Social Distance Width of two polygons $P_1,P_2$ as a special case in which $R$ is the boundary of $P_1$ and $B$ is the boundary of $P_2$; i.e., $SDW(P_1,P_2)=SDW(\partial P_1,\partial P_2)$. 
Similarly, the Social Distance Width of a (single) polygon $P$ is $SDW(P)=SDW(\partial P,\partial P)$. We have

\begin{theorem}
    The social distance width of a polygon $P$ of $n$ vertices can be computed in$O(n^2)$ time.
    
\end{theorem}
The notion of SDW of a polygon is possibly related to other characteristics of polygons, such as fatness. Intuitively, if the polygon $P$ is fat under standard definitions, then the SDW of $P$ will be large. However, the exact connection is yet unclear (see \Cref{fig:fatness}), and we leave open the question of what exactly is the relation between the two definitions.
\begin{figure}[h]
    \centering
    \includegraphics{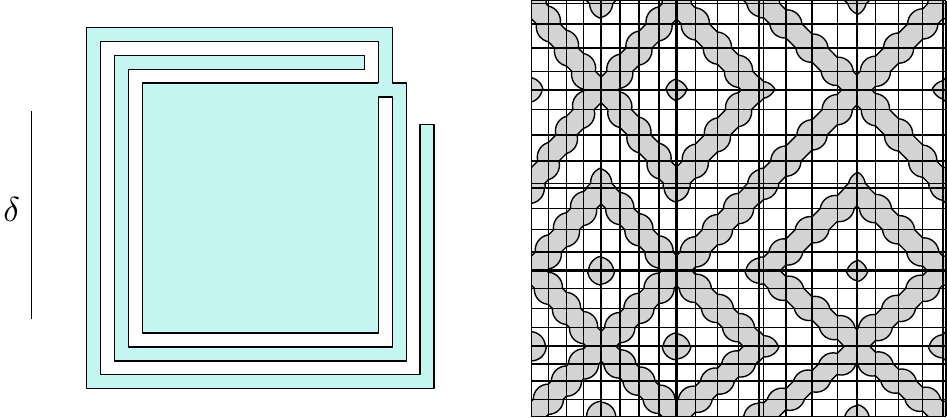}
    \caption{A $\delta$-fat polygon $P$, and the free space diagram showing that $SDW(P)\ll \delta$.}
    \label{fig:fatness}
\end{figure}

\vspace{10pt}
\noindent
We now show that in special cases, the SDW can be computed in linear-time.

\subsection{Social distancing in a skinny polygon (or a tree)}\label{sec:polygon-tree}
In this section we consider the case in which the shared domain of Red and Blue is a tree $T$, and the distance is the shortest-path distance in the tree (the distance between vertices $u$ and $v$ denoted $|uv|$). Red moves around $T$ in a depth-first fashion: there is no start and end point, it keeps moving ad infinitum. In particular, if $T$ is embedded in the plane, the motion is the limiting case of moving around the boundary of an infinitesimally thin simple polygon, and the distance is the geodesic distance inside the polygon.

\begin{theorem}
	Let $T$ be a tree with $n$ vertices, embedded in the plane, and $R$ be a traversal of $T$ in a depth-first fashion. There exists an $O(n)$-time algorithm to find a path $B$ in $T$ that maximizes $SDW(R,B)$ under the geodesic distance.
\end{theorem}

\noindent\textit{Proof.} 	We start with the case when $T$ is a star (\Cref{fig:tree}, left). Let $r$ be the root of the star and let $|ra|\ge|rb|\ge|rc|$ be the 3 largest distances from $r$ to the leaves (i.e., the distance to the root from all other leaves is at most $|rc|$). Assume that the leaves $a,b,c$ are encountered in this order as Red moves around $T$ (this assumption is w.l.o.g., since the other orders are handled similarly); we call $r$ and $|rc|$ the \emph{2-outlier center} and \emph{radius} of $T$ because allowing 2 outliers, $|rc|$ is the smallest radius to cover $T$ with a disk centered at a vertex of the tree. Now, on the one hand, Blue can maintain distance $|rc|$ from Red: when Red is at $a$, Blue is at $c$; when Red is at $b$, Blue moves to $a$; when Red is at $c$, Blue moves to $b$; the minimum distance of $|rc|$ is achieved when Blue is at $c$. On the other hand, the distance must be at least $|rc|$ at some point, since Blue cannot sit at $a$ or at $b$ all the time, and, while Blue moves from $a$ to $b$ through $r$, Red must be somewhere else (other than $a$ or $b$). 
		%\begin{figure}[h]
	\begin{wrapfigure}{R}{0.4\textwidth}
		\centering\includegraphics{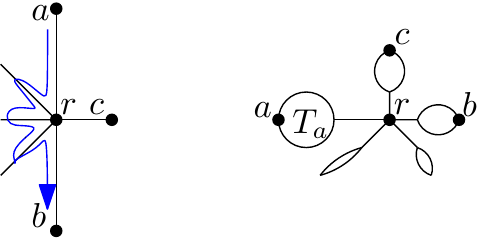}\caption{\small Red is at $c$ while Blue moves between $a$ and $b$. Left: A star. Right: An arbitrary tree.}\label{fig:tree}
	\end{wrapfigure}
	%\end{figure}

	We now consider an arbitrary tree $T$. Let $r\in T$ be a vertex. Removal of $r$ disconnects $T$ into several trees; for a vertex $v\ne r$ of $T$ let $T_v\ni v$ be the subtree of $v$. Let $a$ be the vertex of $T$ furthest from $r$, let $b$ be the vertex of $T\setminus T_a$ furthest from $r$, and let $c$ be the vertex of $T\setminus T_a\setminus T_b$ furthest from $r$ (\Cref{fig:tree}, right). Call $|rc|$ the \emph{2-outlier radius} of $r$, and assume $r^*$ is the vertex whose 2-outlier radius is the largest. As in a star, Blue can maintain the distance of $|r^{*}c|$ from Red by cycling among $a,b,c$ "one step behind" Red. Also as in a star, a larger distance cannot be maintained because, again, Blue has to pass through $r^*$ on its way from $a$ to $b$, and the best moment to do so is when Red is at $c$.
	
	To find $r^*$ in linear time, note that $ab$ is a diameter of $T$: it is a longest simple path in $T$. All diameters of a tree intersect because if two diameters $uv,u'v'$ do not intersect, then there exist vertices $w\in uv,w'\in u'v'$ that connect the two diameters and the distance from each of $w,w'$ to one of the endpoints of its diameter is at least half the diameter, implying that the distance between these endpoints is strictly larger than the diameter (\Cref{fig:diam}, left). Moreover, since the tree has no cycles, the intersection of all its diameters is a path $\pi$ in $T$. Notice that the distance from any diameter endpoint to the closest point on $\pi$ is the same (call it $d$). We claim for any point $r'$ not on $\pi$, the point $r$ on $\pi$ closest to it has a larger 2-outlier radius, and thus $r^*$ may be found on $\pi$. Indeed, let $T_{r}$ be the tree that contains $r$ after removing $r'$, and let $v$ be the farthest point from $r'$ not in $T_{r}$. Since $|r'v|$ is strictly smaller than the distance between $r$ and the closest diameter endpoint (See~\Cref{fig:diam}, right), we get that the 2-outlier radius of $r'$ is at most $|r'v|$. On the other hand, the 2-outlier radius of $r$ is at least $|rv|$, and clearly $|rv|>|r'v|$.
	
	\begin{figure}[h]
	\centering\includegraphics[scale=0.9]{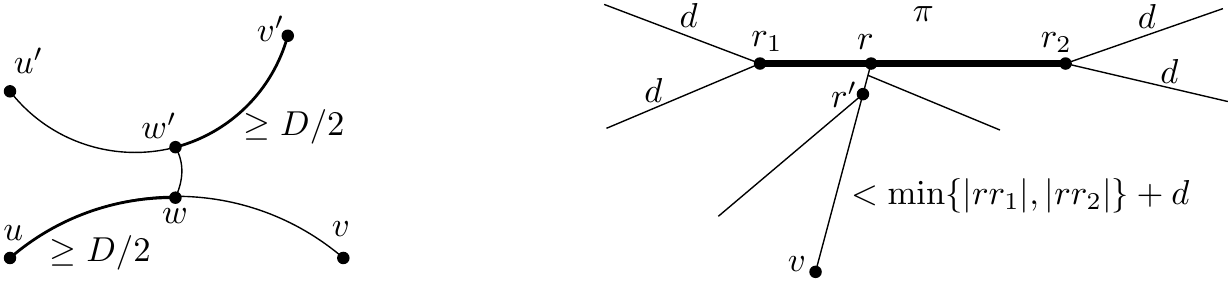}\caption{Left: If $|uv|=|u'v'|=D$, then $|uww'v'|>D$. Right: The distance from any diameter endpoint to the closest endpoint of $\pi$ (thick) is the $d$, for otherwise one of the diameters is longer than another. The 2-outlier radius of $r$ is $d$, while the 2-outlier radius of $r'$ cannot exceed $d$.}\label{fig:diam}
	\end{figure}
	
	We thus compute a diameter $ab$ (linear time) and pick the vertex with the largest 2-outlier radius on the diameter as $r^{*}$ by checking the vertices one by one. As we check consecutive vertices on $ab$, the distances $|ra|$ and $|rb|$ are updated trivially, and the subtrees $T\setminus T_a\setminus T_b$ are pairwise-disjoint for different vertices $r$ along the diameter; thus the longest paths in all the subtrees can be computed in total linear time. 
\qed

\vspace{2pt}

%A discussion on more variants of the problem can be found in \Cref{apx:conclusion}.
The discussion on this special case of thin polygons, leads us to investigating the more general case of Social Distance Width of two graphs, which we discuss next.%in \cref{sec:graphs}.

\section{Social distancing on graphs}\label{sec:graphs}

In this section we consider social distancing measures for graphs. We focus on undirected graphs, but our results readily extend to directed graphs. A \emph{geometric} graph is a connected graph with edges that are straight line segments embedded in the plane. We consider geometric graphs to be automatically weighted, with Euclidean edge weights. 

We will call non-geometric graphs simply as abstract graphs, and unless otherwise mentioned, we give each edge of an abstract graph a length of one. Since we are interested in formulating meaningful distance measures, we assume that an all-pairs shortest path preprocessing has been done on $G$, and pairwise distances are available in $O(1)$ time. Let $G=(\mathcal{V},E)$, $n=\vert \mathcal{V}\vert$ and $m = \vert E \vert$. By a path of length $k$ in an abstract graph $G=(\mathcal{V},E)$ we mean a sequence $P$ of $k$ vertices: $P = \{ p_{1}, p_{2}, \cdots, p_{k}\}\subseteq \mathcal{V}$, with edges $p_{i}p_{i+1}\in E$ for all $1 \leq i \leq k-1$, and we write $\vert P \vert =k$. Note that we allow vertices to be repeated (that is, we don't distinguish between a path and a walk). 

For abstract graphs, we assume that Red and Blue both move on the same graph, as there is no ambient metric defining the distance between Red and Blue. For geometric graphs, however, we will also consider the case when Red and Blue move on different graphs or in the embedding space, as the distance between them is still well-defined. 

\smallskip
\subsection{Abstract Graphs}
We state our result for Blue staying away from Red on an abstract graph.
\begin{theorem}[Blue distancing from Red, abstract]
Assume Red travels on a known path $R$ of length $k$ in an abstract graph $G$, and Blue can travel anywhere with a speed $s \geq 0$ times that of Red, for some integer $s$. There exists a decision algorithm for the Blue distancing from Red-on-a-mission problem that
 \begin{itemize}     
 \item for $s \in \{0, 1, n-1\}$,  runs in time $O(km)$, 
    \item for arbitrary $ 1 < s < n-1$, runs in time $O(nk \min(d^{s},n))$, where $d$ is the maximum degree of any vertex in $G$.
\end{itemize}
\end{theorem}

\begin{proof}
Take $k$ copies $G_i$ of $G$, for $1 \leq i \leq k$; let $(v,i)$ be the vertex corresponding to vertex $v$ in the $i$th copy $G_i$. Construct a directed graph $\tilde{G}$ by connecting these copies as follows (it may be helpful to imagine these copies ``stacked'' on top of each other):
\begin{itemize}
    \item For every $v \in V$ and every $1 \leq i \leq k-1$, add a directed edge from $(v,i)$ to $(v,i+1)$.
    \item If $s=0$, add no more edges. If $s=1$, add an edge between the vertex $(v,i)$ and vertex $(w,i+1)$ for all neighbors $w$ of $v$ in $G$, for all $(v,i)$, $1 \leq i \leq k-1$.
    \item If $1 < s < n-1$, add an edge between between the vertex $(v,i)$ and vertex $(w,i+1)$ for all vertices $w$ such that the distance between $w$ and $v$ in $G$ is at most $s$, for all $(v,i)$, $1 \leq i \leq k-1$. Note that there are at most $d^s$ many such vertices, where $d$ is the maximum degree of a vertex in $G$.
    \item if $s = n-1$, add an edge from vertex $(v,i)$ to all vertices in $G_{i+1}$.
\end{itemize}

Let $R=\{r_1,r_2,\cdots,r_k\}$ be the path of Red. Once $\tilde{G}$ is constructed, delete the vertex $(r_i,i)$ and all its neighbors from $G_i$ if Blue is looking to stay distance at least one away from Red. If Blue instead aims to stay $\delta$ away, delete the $\delta$ neighborhood of these vertices from each copy. This results in a graph $G'$. If Blue's starting position $b$ is given, run BFS/DFS to see if $b$ can reach any vertex in $G_k$, and return yes or no accordingly. If blue can start from any vertex, then check if any vertex in $G_1$ can reach some vertex in $G_k$, and answer yes or no accordingly. The runtime depends on the complexity of $\tilde{G}$, which is $O(mk)$ if $s \in \{0,1,n-1\}$, and $O(nk \min(d^{s},n))$ otherwise, which completes the proof.
\end{proof}

\subsection{Geometric Graphs} Now we consider geometric graphs. Let $R:[0,\ell] \rightarrow \mathbb{R}^2$ be a polygonal curve in $\mathbb{R}^2$, consisting of $\ell$ line segments, with the $i$th line segment $R_{i}= R \vert_{[i-1,i]}$ where $1\leq i \leq \ell$. We also assume that each segment $R_i$ is parameterized naturally by $R(i+\lambda)=(1-\lambda)R(i) + \lambda R(i+1)$. The curve $R$ is assumed to be known: this is how Red is traveling. On the other hand, we are also given a geometric graph $G$ in which Blue is restricted to travel, and the problem is to determine if there is a path $P \in G$ (a path in $G$ is the polygonal curve formed by the edges between the start and end points of $P$), such that $SDW(R,P) \geq \delta$, where the SDW between $R$ and $P$ is defined as in the continuous SDW between polygonal curves in Section 2. Note that by adopting this definition we are considering the Euclidean distance between Red and Blue: if one instead considers geodesic distance, Red will also need to be restricted to travel in $G$. We remark on this setting later. 
\begin{theorem}[Blue distancing from Red, geometric]
The decision problem, given $R$ (the path of Red) and $G$ (the graph of Blue), and a distance $\delta$, can be solved in time $O(k m)$, where $k$ is the length of $R$ (number of vertices) and $m$ is the number of edges in $G$.
\end{theorem}

\begin{proof}

We first define the free space surface for our problem.
Consider an edge $e_{i,j}=(v_i,v_j) \in E$, and let $\F_\delta^{i,j}$ denote the $\delta$-free space of $e_{i,j}$ and $R$ (note that $e_{i,j}: [0,1] \rightarrow \mathbb{R}^2$ is a polygonal curve of length 1).
Analogous to \cite{AERW03}, we glue the corresponding free-space diagrams of any two edges in $E$ that share a common vertex, along the boundary of the diagram that corresponds to this vertex. Doing so for all edges $e_{i,j} \in E$
yields the free space surface $\mathcal{S}$.

As in \cite{AERW03}, we get that there exists a path $P$ in $G$ with $SDW(R,P) \ge \delta$ if and only if there exists a $y$-monotone path through the $\delta$-free space in $\mathcal{S}$ between a point $(0,s)$ in some $\F_\delta^{i,j}$, to a point $(k,t)$ in some $\F_\delta^{i',j'}$. Such a path corresponds to a continuous surjective traversal of $R$ from $R(0)$ to $R(k)$, and some continuous non-surjective traversal $P$ of $G$ that starts on a point $s$ on the edge $e_{i,j}$ and ends on a point $t$ on $e_{i',j'}$. By similar arguments, such a path can be found in $O(k m)$ time by computing the reachability diagram as in \cite{AG95}.

\end{proof}

\noindent\textbf{SDW of a Graph:} We now move to defining the social distance width of a graph that we briefly mentioned above. A curve on a geometric graph corresponding to a surjective (non-surjective) map will be called a traversal (partial traversal) of the graph. We note that traversals may need to backtrack, and we allow partial traversals to do so too.

Define, for two geometric graphs $H$ and $G$,
$ SDW(H,G) = \sup_{h,g} \min_{t \in [0,1]} d(h(t),g(t)),$
where $h$ is a traversal of $H$ and $g$ is a \textit{partial} traversal of $G$. The setting is that Red is going about its business on $H$, and must traverse it completely in some order, while Blue is only trying to stay away, and is restricted to be on $G$. Finally, we define for a graph $G$, its social distance width, as $SDW(G)=SDW(G,G)$ where again, we are free to choose either the Euclidean or the geodesic version. 

\begin{theorem}
For two geometric graphs $H$ and $G$, the Euclidean $SDW(H,G)$ can be computed in time $O(m_h m_g)$, where $m_h$ (resp. $m_g$) denotes the number of edges in the graph $H$ (resp. $G$).
\end{theorem}

\begin{proof}
It is enough to show how to compute $SDW(H,G)$; $SDW(G)$ is then obtained by putting $H=G$. To this end, we construct the free space surface in a manner similar to that of the Red-on-a-mission case described previously. For every edge $e=(u,v) \in H$ and $f=(x,y) \in G$, we construct a free space cell $C_{e,f}$, which can be thought of as a subset of $[0,1]^{2}$. For edges $e$ and $e'$ sharing a vertex $v$, we glue $C_{e,f}$ and $C_{e',f}$ along their right and left edges, respectively, which correspond to $C
_{v,f}$. In this way we obtain a cell complex in three dimensions, with faces corresponding to cells $C_{e,f}$, edges corresponding to $C_{u,f}$ (or $C_{v,f}$, $C_{e,x}$, or $C_{e,y}$), and vertices corresponding to $C_{u,x}$, etc. The rest of the proof follows the lines of~\cite{AERW03}.
\end{proof}

The above result shows that when $d$ denotes the Euclidean distance, $SDW(H,G)$ can be computed in time $O(m_hm_g)$. The following observation shows that even the geodesic case behaves nicely, as the free space is convex.

\begin{lemma}
For any $L>0$, the free space (defined to be the points for which the geodesic distance is at least $L$) inside a cell of the free-space diagram is a convex polygon.
\end{lemma}

\begin{proof}
Consider 2 edges, $e=(v_1,v_2)$ and $f=(u_1,u_2)$ in an embedded PSLG (planar straight line graph), with coordinates $x$ on $e$ such that $x(v_1)=0$, $y$ on $f$ such that $y(u_1)=0$.

Let $d_1,d_2,d_3$ and $d_4$ be the distances between $v_1 u_2$, $v_1u_1$, $v_2u_2$ and $v_2u_1$, respectively.  Then, the geodesic distance between a point at position $x$ on $e$ and a point at position $y$ on $f$ is simply 

\[ d(x,y)=\min(x+d_2+y, x+d_1+|f|-y, |e|-x+d_4+y, |e|-x+d_3+|f|-y), \]

which is a minimum of four functions that are linear in $(x,y)$.
Thus, $d(x,y)$ is concave, piecewise-linear, and the locus of points $(x,y)$ for which $d(x,y) \geq L$ is a convex polygon, for any $L$.

\end{proof}

As a result, we obtain an algorithm with the same running time for geodesic distance.

	%\doubleBlind{
	\section*{Acknowledgements}
	We thank the many participants of the Stony Brook CG Group, where discussions about geometric social distancing problems originated in Spring 2020, as the COVID-19 crisis expanded worldwide. We would also like to thank Gaurish Telang for helping to solve a system of non-linear equations.
	%}
	
	%\newpage
	
	\bibliographystyle{alphaurlinit}
	\bibliography{refs}
	
	\appendix

\end{document}